\documentclass[12pt]{article}
 
\usepackage{amsmath,amssymb,amsthm}
\usepackage{amsfonts}
\newtheorem{definition}{Definition}
\newtheorem{theorem}{Theorem}

\usepackage{graphicx}

\usepackage{amssymb}
\usepackage{url}

\usepackage[cp1250]{inputenc}
\newcommand{\Bem}[1]{}
\newcommand{\rys}[1]{}
 
\newcommand{\affil}[1]{\\ #1 \\}

\begin{document}
\title{Traditional PageRank versus Network Capacity Bound 
}
\author{
MIECZYS{\L}AW A.K{\L}OPOTEK, S{\L}AWOMIR T.WIERZCHOÑ  \affil{Institute of Computer Science of   Polish Academy of Sciences\\ Warszawa, Poland}  
ROBERT A. K{\L}OPOTEK\affil{Stefan Cardinal Wyszy\'{n}ski University in Warsaw\\ Warszawa, Poland}  
EL¯BIETA A. K{\L}OPOTEK \affil{mBank, Warszawa, Poland}
}

\maketitle

\begin{abstract}
\noindent
In a former paper \cite{Bipartite:2012} we  
simplified the proof of a theorem on
personalized random walk that is fundamental
to graph nodes clustering
and   generalized it to bipartite graphs
for a specific case where the proobability of random jump was proprtional to the number of links of "personally prefereed" nodes.
In this paper we turn to the more complex issue of graphs in which the random jump follows uniform distribution.
\end{abstract}
 
\section{Introduction}

The PageRank is widely used as a (main or supplementary) measure of importance of a web page since its publication  in \cite{Page:1999}. 
Subsequently the idea was explored with respect to   methods of computation \cite{Berkhin:2005}, application areas   (Web page ranking, client and seller ranking, clustering, classification of web pages, word sense disambiguation, spam detection, detection of dead pages etc.)
and application related variations (personalized PageRank, topical PageRank, Ranking with Back-step, Query-Dependent PageRank, Lazy Walk Pagerank etc.), \cite{bibliography:PageRank}.

The traditional PageRank reflects the probability 
that a random walker reaches a given webpage.
The walker, upon entering a webpage, follows with uniform probability one of the outgoing edges unless he gets bored or there are no outgoing edges. 
In this case he jumps to any  web page with uniform probability. 

As already mentioned, one of the application areas 
of PageRank is creation of new clustering methods especially for graphs, including undirected\footnote{Unoriented graphs
have multiple applications as means to represent   relationships
spanned by a network of friends,
telecommunication infrastructure
or street network of a city}  graphs in which we are interested in this paper. 
One of clues for clustering of graphs assumes that a good cluster has low probability   to be left by a random walker. Though the concept seems to be plausible, is has been investigated theoretically only for a very special case of a random walker (different from the traditional walker),  performing the "boring  jump"  with probability being   proportional to
 the number of incident edges (and not uniformy)  -- see e.g. \cite{Chung:2011,Bipartite:2012}.

\Bem{
\begin{figure}
   \centering
   \includegraphics[width=0.9\textwidth]{\rys{}siecv.jpg}\\
\caption{An unoriented network}
\label{siecv}
\end{figure}
}

In this paper  we will make an attempt to
extend this result to the case when the "boring jump" is performed uniformly (as in case of traditional walker)  (Section 2) and to generalize it to bipartite graphs (Section 3).

PageRank computation for bipartite graphs was investigated already in the past in the context of social networks, 
  e.g. when concerning mutual evaluations of  students and lecturers \cite{Link:2011}, reviewers and movies in a movie recommender systems, or authors and papers in scientific literature or queries and URLs in query logs   \cite{DLK09}, or performing image tagging \cite{Bauckhage:2008}. 
As pointed at in \cite{Bipartite:2012},   
 the bipartite graphs have explicitly a periodic structure while PageRank aims at graph aperiodicy. 
 Therefore a suitable generalization of PageRank
to bipartite structure is needed and we will follow here the proposals made in \cite{Bipartite:2012}.
 
\Bem{
\begin{figure}
   \centering
   \includegraphics[width=0.9\textwidth]{\rys{}wevdrowcy.jpg}\\
\caption{Random walker interpretation of PageRank}
\label{wevdrowcy}
\end{figure}
}

\section{Traditional  PageRank}

One of the many interpretations of PageRank views it as 
 the probability that a knowledgeable (knowing addresses of all the web pages) but  mindless (choosing next page to visit without regard to any content hints) random walker  will encounter a given Web page. 
 So  upon entering a particular web page, if it has no outgoing links, the walker jumps to any Web page with uniform probability. If there are outgoing links, he  chooses with uniform  probability one of the outgoing links and goes to the selected web page, unless he gets bored.
If he gets bored (which may happen with a fixed probability
$\zeta$ on any page), he jumps to any Web page with uniform probability.

One of the modifications of this behavior (called personalized PageRank) was 
  a mindless page-$u$-fan
random walker who is doing exactly the same, but in case of a jump out of boredom he does not jump to any page, but to the page $u$.\footnote{ If there exists one page-fan for each web page then the PageRank vector  of the knowledgeable walker
is the average of PageRank vectors of all these page-fan walkers} 

Also there are plenty possibilities of other mindless walkers
between these two extremes. For example upon being bored
the walker can jump to a page from a set $U$ with a uniform probability or with probability proportional to the out-degree of the pages. An unacquainted reader is warmly referred to \cite{LM06} for a  detailed treatment of these topics.

\Bem{
\begin{figure}
   \centering
\includegraphics[width=0.9\textwidth]{\rys{}boundaryVolume.jpg}\\
\caption{A preferred group of pages of a random walker}
\label{boundaryVolume}
\end{figure}
 }

Let us recall formalization of these concepts.  
 With $\mathbf{r}$ we will denote a (column) vector of ranks: $r_j$ will mean the PageRank of page $j$. All elements of $\mathbf{r}$ are non-negative and their sum equals 1.

Let $\mathbf{P} = [p_{ij}]$ be a matrix such that if there is a link from page $j$ to page $i$, then $p_{i,j}=\frac{1}{outdeg(j)}$,
where $outdeg(j)$ is the out-degree of node $j$\footnote{
For some versions of PageRank, like TrustRank
$p_{i,j}$ would differ from $\frac{1}{outdeg(j)}$
giving preferences to some outgoing links over the other.
We are not interested in such considerations here.}.
In other words, $\mathbf{P}$ is column-stochastic matrix satisfying $\sum_i p_{ij} = 1$ for each column $j$. If a node had an out-degree equal 0, then prior to construction of $\mathbf{P}$ the node is replaced by one with edges outgoing to all other nodes of the network.

Under these circumstances we have

\begin{equation}
\mathbf{r}=(1-\zeta){\cdot} \mathbf{P}{\cdot}\mathbf{r} + \zeta {\cdot} \mathbf{s} \label{eqPageRankDef}
\end{equation}

\noindent where $\mathbf{s}$ is the so-called ``initial'' probability distribution (i.e. a column vector with non-negative elements summing up to 1) that is also interpreted as a vector of Web page preferences.\footnote{
We will denote the solution to the equation 
(\ref{eqPageRankDef}) with 
$\mathbf{r}^{(t)}(\mathbf{P},\mathbf{s},\zeta)$. 
}

For a knowledgeable walker for each node $j$ of the network  $s_j=\frac{1}{|N|}$, where $|N| $ is the cardinality of the set of nodes $N$ constituting the network. For a page-$u$-fan we have $s_u=1$, and $s_j=0$ for any other page $j\ne u$. 
For  a uniform-set-$U$-fan\footnote{We will call the set $U$ "fan-pages" or "fan-set" or "fan-nodes"}  we get

\[
s_j = \left\{
\begin{array}{ll}
\displaystyle\frac{1}{|U|} & \textrm{if $j \in U$}\vspace{0.1cm}\\
0 & \textrm{otherwise}
\end{array} \right.,\ j = 1, \dots |N|
\label{eq-uni}
\]

\noindent and for a hub-page-preferring-set-$U$-fan we obtain

\begin{equation}
s_j = \left\{
\begin{array}{ll}
\displaystyle\frac{outdeg(j)}{\sum_{k\in U} outdeg(k)} & \textrm{if $j \in U$}\vspace{0.1cm}\\
0 & \textrm{otherwise}
\end{array} \right.,\ j = 1, \dots |N|
\label{eq-hub}
\end{equation}

The former case is the topic of this paper, the second was considered in our former paper \cite{Bipartite:2012}.

Instead of a random walker model we can view a Web as a pipe-net through which the authority is flowing in discrete time steps.

In single time step a fraction $\zeta$ of the authority of a node $j$ flows into so-called \emph{super-node}, and the fraction $\frac{1-\zeta}{outdeg(j)}$ is send from this node to each of its children in the graph. After the super-node has received authorities from all the nodes, it redistributes the authority to all the nodes in fractions defined in the vector $\mathbf{s}$. Note that the authority circulates lossless (we have a kind of a closed loop here).

Beside this, as was proven in many papers, we have to do here with a self-stabilizing process. Starting with any stochastic vector $\mathbf{r}^{(0)}$ and applying the operation

$$\mathbf{r}^{(n+1)}=(1-\zeta){\cdot}\mathbf{P}{\cdot}\mathbf{r}^{(n)}+\zeta{\cdot}s$$

\noindent the series $\{\mathbf{r}^{(n)}\}$ will converge to
$\mathbf{r}$ being the solution of the equation (\ref{eqPageRankDef}) (i.e. to the main eigenvector corresponding to eigenvalue 1).

Subsequently let us consider only connected graphs (one-component graphs) with symmetric links, i.e. undirected graphs. Hence for each node $j$ the relationships between in- and out-degrees are:  $$ indeg(j)=outdeg(j)=deg(j)$$ Let us pose the question: how is the PageRank of $U$-set pages related to the PageRank of other pages (that is those pages where there are no jumps out of being bored)?

In a former paper we have proven \cite{Bipartite:2012}
\begin{theorem} \label{thPRlimtPreferential}
For the preferential personalized PageRank we have

$$p_o\zeta\le (1-\zeta)\frac{|\partial(U)|}{Vol(U)}$$

\noindent where $\partial(U)$ is the set of edges leading from $U$ to the nodes outside of $U$ (the so-called ``edge boundary of $U$''), hence $|\partial(U)|$ is the cardinality of the boundary, and $Vol(U)$, called volume or capacity of $U$ is the sum of out-degrees of all nodes from $U$.
\end{theorem}

Let us discuss now a uniform-set-$U$-fan  defined in equation (\ref{eq-uni}).  
Let us now turn to the situation where $U$ is only a proper subset of $N$, and assume that

\begin{equation}
r_j^{(t)} = \left\{
\begin{array}{ll}
\displaystyle\frac{1}{|U||} & \textrm{if $j \in U$}\vspace{0.1cm}\\
0 & \textrm{otherwise}
\end{array} \right.,\ j = 1, \dots |N|
\label{eq-uniA1}
\end{equation}

\noindent in a moment $t$. To find the distribution $\mathbf{r}^{(t')}$ for $t'>t$ we state that if in none of the links the passing amount of authority will exceed $$\gamma=(1-\zeta)\frac{1}{|U| \min_{k\in U} deg(k)}$$
then at any later time point $t'>t$ the inequality
$r_{j}^{(t')}\le deg(j)*\gamma + \frac{\zeta }{|U|}  $ holds at any node $j \in U$.

To justify this statement note that if a node $j\not\in U$ gets via links $l_{j,1},...,l_{j,deg(j)}$ the authority amounting to

$$a_{l_{j,1}}\le \gamma,...,a_{l_{j,deg(j)}}\le \gamma$$

\noindent then it accumulates

\[\mathfrak{a}_j = \sum_{k=1}^{deg(j)} a_{j,k} \le \gamma{\cdot}deg(j)\]

\noindent of total authority, and in the next time step the following amount of authority flows out through each of these links:

$$(1-\zeta)\frac{\mathfrak{a}_j}{deg(j)}
\le \gamma(1-\zeta)\le \gamma$$

If a node $j \in U$ gets via incoming links $l_{j,1},...,l_{j,deg(j)}$ the authority amounting to
$a_{l_{j,1}}\le \gamma,...,a_{l_{j,deg(j)}}\le \gamma$
then, due to the authority obtained from the super-node
equal to $ \mathfrak{b}_j = \zeta \frac{1}{|U| } \le deg(j)\gamma\frac{\zeta}{1-\zeta}$, in the next step through each link the authority amounting to

$$(1-\zeta)\frac{\mathfrak{a}_j}{deg(j)}+(1-\zeta)\frac{\mathfrak{b}_j}{deg(j)}
\le \gamma(1-\zeta)+ \gamma\frac{\zeta}{1-\zeta}(1-\zeta)
$$ $$=\gamma(1-\zeta)+\gamma\zeta = \gamma$$
flows out.

So if already at time point $t$ the authority flowing out through any link from any node did not exceed $\gamma$, then this property will hold
(by induction) forever, especially for the equation solution $\mathbf{r}$ which is unique. 
 

Now let us ask: ``How much authority from outside of $U$
can flow into $U$ via super-node at the point of stability?''
Let us denote by $p_o$ the total mass of authority contained in all the nodes outside of $U$. Then our question concerns the quantity $p_o\zeta$. 
We  claim  that  

\begin{theorem} \label{thPRlimit}
For the uniform personalized PageRank we have

$$p_o\zeta\le 
=(1-\zeta)\frac{|\partial(U)| }{|U| \min_{k\in U} deg(k)}
$$
\noindent 
\end{theorem}

\begin{proof}
Let us notice first that, due to the closed loop of authority circulation, the amount of authority flowing into $U$ from the nodes belonging to the set $\overline{U} = N \backslash U$ must be identical with the amount flowing out of $U$ to the nodes~in~$\overline{U}$.

But from $U$ only that portion of authority flows out that flows out through the boundary of $U$ because no authority leaves $U$ via super-node (it returns from there immediately). As at most the amount $\gamma |\partial(U)|$ leaves $U$, then

\[p_o\zeta\le \gamma |\partial(U)| 
=(1-\zeta)\frac{1}{|U| min_{k\in U} deg(k)}|\partial(U)| 
=(1-\zeta)\frac{|\partial(U)| }{|U| min_{k\in U} deg(k)}\]
\end{proof}

 \Bem{
\begin{figure}
   \centering
   \includegraphics[width=0.9\textwidth]{\rys{}dwudzielny.jpg}\\
\caption{An example of a bipartite graph}
\label{dwudzielny}
\end{figure}
}

When you compare the above two theorems \ref{thPRlimtPreferential} and \ref{thPRlimit}, you will see immediately that the bound in case of "preferential" theorem \ref{thPRlimtPreferential} is lower than in case of "uniform" theorem \ref{thPRlimit}. 

If we look more broadly at the  $s$ vector with 
$s_j>0 \ \forall_{j\in U}$
and  $s_j=0 \ \forall_{j\not\in U}$, we will derive immediately by analogy the relation 

\begin{theorem} \label{thPRlimitGeneral}
For the   personalized PageRank
with arbitrary $s$ vector such that 
$s_j>0 \ \forall_{j\in U}$
and  $s_j=0 \ \forall_{j\not\in U}$
 we have

$$p_o\zeta\le 
=(1-\zeta)\frac{|\partial(U)| }{  \min_{k\in U} \frac{deg(k)}{s_k}}
$$
\noindent 
\end{theorem}

\section{Variants of the theorems}

In this section 
our attention is concentrated
on some versions of PageRank related to diverse methods of random walk that have distinct semantic connotations. 
  Each of the four versions mentioned below  represents semantically different behavior of the surfers and hence the respective PageRank has different commercial values as placement if e.g. advertisement is concerned.

The previously considered traditional PageRank represents a random walker
travelling at uniform pace through the Internet. The lazy walk PageRank reflects the distribution of walkers that may stay longer at a given page than just a single unit of time. The generalized lazy walk PageRank allows for simulation of leaning either towards jumping or ``reading'' of a Web page. The random walk with backstep PageRank refers to random walkers that may withdraw from a step forward if they find the page uninteresting. Two brands have been studied: one going only one step backwards before going forward and one with deeper backsteps.

Though these behaviors seem to be semantically quite distinct,
it will be demonstrated below 
 that mathematically they can be reduced to a single form but by differentiating the boring factor. This means that in fact we need only one version of any PageRank related algorithm for computation in each case, in particular the ones related to clustering issues. On the other hand also theorems, in our case concerning authority flow limits, can be easily transferred between the models This is good news. The bad one is that by observing the PageRank vector alone we cannot decide which type of random walk we have to do with. Even the knowledge of the boring factor is insufficient to distinguish between generalized random walk and random walk with backstep.

\subsection{Lazy Random Walk PageRank}

A variant of PageRank, 
so-called \emph{lazy-random-walk-PageRank}
was described e.g. by \cite{Chung:2008x}.
It differs from the traditional PageRank
in that the random walker
before choosing the next page to visit
he fist tosses a coin and upon heads 
he visits the next page
and upon tails he stays in the very same node of the network.

Let us recall here a known relationship.

For the lazy walker PageRank we have:

\begin{equation}
\mathbf{r}^{(l)}=(1-\zeta){\cdot}\left(0.5 \mathbf{I} + 0.5 \mathbf{P}\right){\cdot} \mathbf{r}^{(l)} + \zeta {\cdot} \mathbf{s} \label{eqLazyPageRankDef}
\end{equation}
where $\mathbf{I}$ is the identity matrix.
\footnote{
We will denote the solution to the equation 
(\ref{eqLazyPageRankDef}) with 
$\mathbf{r}^{(l)}(\mathbf{P},\mathbf{s},\zeta)$. 
}

One can easily guess relation to the traditional PageRank. 
Let us transform: First multiply by 2 
$$
2\mathbf{r}^{(l)} =(1-\zeta){\cdot}\left( \mathbf{I}+ \mathbf{P}\right){\cdot} \mathbf{r}^{(l)} + 2\zeta {\cdot} \mathbf{s} 
$$
Now subtract 
$$
2\mathbf{r}^{(l)}-(1-\zeta) \mathbf{I}\mathbf{r}^{(l)} =(1-\zeta){\cdot}\left( \mathbf{P}\right){\cdot} \mathbf{r}^{(l)} + 2\zeta {\cdot} \mathbf{s}  
$$

$$
(1+\zeta) \mathbf{r}^{(l)}  =(1-\zeta){\cdot}\left( \mathbf{P}\right){\cdot} \mathbf{r}^{(l)} + 2\zeta {\cdot} \mathbf{s}  
$$
and divide by $(1+\zeta)$
$$
 \mathbf{r}^{(l)}  =\frac{1-\zeta}{1+\zeta}{\cdot}\left( \mathbf{P}\right){\cdot} \mathbf{r}^{(l)} + \frac{2\zeta}{1+\zeta} {\cdot} \mathbf{s}  
$$

This means that 
 $\mathbf{r}^{(l)} $   for $\zeta$ 
is the same as 
 $\mathbf{r}^{(t)} $   for $ \frac{2\zeta}{1+\zeta} $ 
($\mathbf{r}^{(l)}(\mathbf{P},\mathbf{s},\zeta)=
 \mathbf{r}^{(t)}(\mathbf{P},\mathbf{s},\frac{2\zeta}{1+\zeta})$)

\Bem{
zeta_t=  
$$p_o\zeta\le (1-\zeta)\frac{|\partial(U)|}{Vol(U)}$$
$$p_o\frac{2\zeta}{1+\zeta}\le (1-\frac{2\zeta}{1+\zeta})\frac{|\partial(U)|}{Vol(U)}$$
$$p_o {2\zeta}\le ({1+\zeta}- {2\zeta} )\frac{|\partial(U)|}{Vol(U)}$$
$$p_o {2\zeta}\le ({1-\zeta} )\frac{|\partial(U)|}{Vol(U)}$$
$$p_o {\zeta}\le \frac{{1-\zeta} }{2}\frac{|\partial(U)|}{Vol(U)}$$
}

Under these circumstances we have 
\begin{theorem} \label{thPRlimtPreferentialLazy}
For the preferential lazy personalized PageRank we have
$$p_o {\zeta}\le \frac{{1-\zeta} }{2}\frac{|\partial(U)|}{Vol(U)}$$
\noindent 
\end{theorem}

\begin{theorem} \label{thPRlimitLazy}
For the uniform lazy personalized PageRank we have

$$p_o{\zeta}\le \frac{{1-\zeta} }{2}\frac{|\partial(U)| }{|U| \min_{k\in U} deg(k)}
$$
\noindent 
\end{theorem}

\subsection{Generalized lazy random Walk}

We can generalize this behavior 
to 
\emph{generalized-lazy-random-walk-PageRank}
by introducing the 
laziness degree $\lambda$.
It means that upon tossing the coin is not fair:
probability of tails is $\lambda$
(and heads $1-\lambda$).

For the generalized lazy walker PageRank we have:

\begin{equation}
\mathbf{r}^{(g)}=(1-\zeta){\cdot}\left(\lambda \mathbf{I} + (1-\lambda) \mathbf{P}\right){\cdot} \mathbf{r}^{(g)} + \zeta {\cdot} \mathbf{s} \label{eqGenLazyPageRankDef}
\end{equation}
where $\mathbf{I}$ is the identity matrix.
\footnote{
We will denote the solution to the equation 
(\ref{eqGenLazyPageRankDef}) with 
$\mathbf{r}^{(g)}(\mathbf{P},\mathbf{s},\zeta,\lambda)$. 
}

One can easily guess relation the the traditional PageRank. 
Let us transform: 
$$
\mathbf{r}^{(g)}-(1-\zeta){\cdot} \lambda \mathbf{I} \mathbf{r}^{(g)}
=(1-\zeta){\cdot}  (1-\lambda) \mathbf{P} {\cdot} \mathbf{r}^{(g)} + \zeta {\cdot} \mathbf{s} 
$$

$$
\left(1-(1-\zeta){\cdot} \lambda\right)  \mathbf{r}^{(g)}
=(1-\zeta){\cdot}  (1-\lambda) \mathbf{P} {\cdot} \mathbf{r}^{(g)} + \zeta {\cdot} \mathbf{s} 
$$

$$
\left(1-\lambda+\zeta \lambda\right)  \mathbf{r}^{(g)}
=(1-\zeta){\cdot}  (1-\lambda) \mathbf{P} {\cdot} \mathbf{r}^{(g)} + \zeta {\cdot} \mathbf{s} 
$$

$$
 \mathbf{r}^{(g)}
=\frac{(1-\zeta){\cdot}  (1-\lambda)}
{1-\lambda+\zeta \lambda}
 \mathbf{P} {\cdot} \mathbf{r}^{(g)} 
+ \frac{\zeta}
{1-\lambda+\zeta \lambda}
 {\cdot} \mathbf{s} 
$$

 This means that 
 $\mathbf{r}^{(g)} $   for $\zeta$ 
is the same as 
 $\mathbf{r}^{(t)} $   for $\frac{\zeta}
{1-\lambda+\zeta \lambda} $ 
($\mathbf{r}^{(g)}(\mathbf{P},\mathbf{s},\zeta,\lambda)=
 \mathbf{r}^{(t)}(\mathbf{P},\mathbf{s},\frac{\zeta}
{1-\lambda+\zeta \lambda})$)

\Bem{
zeta_t=  \frac{\zeta}{1-\lambda+\zeta \lambda} 

$$p_o\zeta\le (1-\zeta)\frac{|\partial(U)|}{Vol(U)}$$

$$p_o\frac{\zeta}{1-\lambda+\zeta \lambda}\le (1-\frac{\zeta}{1-\lambda+\zeta \lambda})\frac{|\partial(U)|}{Vol(U)}$$

$$p_o{\zeta}\le ({1-\lambda+\zeta \lambda}-{\zeta} )\frac{|\partial(U)|}{Vol(U)}$$
$$p_o{\zeta}\le (1-\lambda)(1-\zeta)\frac{|\partial(U)|}{Vol(U)}$$
}

Under these circumstances we have 
\begin{theorem} \label{thPRlimtPreferentialLazyGeneralized}
For the preferential generalized lazy personalized PageRank we have
$$p_o {\zeta}\le (1-\lambda)(1-\zeta)\frac{|\partial(U)|}{Vol(U)}$$
\noindent 
\end{theorem}

\begin{theorem} \label{thPRlimitLazyGeneralized}
For the uniform generalized lazy personalized PageRank we have

$$p_o{\zeta}\le (1-\lambda)(1-\zeta)\frac{|\partial(U)| }{|U| \min_{k\in U} deg(k)}
$$
\noindent 
\end{theorem}

\subsection{Random Walk with Backstep}

In \cite{Sydow:2004,Sydow:2004www,Sydow:2005fi} on the other hand so-called
\emph{random-walk-with-backstep-PageRank} (RBS) was introduced.
It differs from the traditional PageRank 
in that a random walker with probability 
$\beta$ chooses to click the backstep button of the 
browser, otherwise 
   (just like 
the ordinary random walker)
that is with   probability $\zeta$ he gets bored and jumps to any page
and with the remaining probability $(1-\beta\zeta)$ he goes to a uniformly chosen child of the page.

\footnote{
We will denote the stationary distribution under this behaviour as 
(\ref{eqGenLazyPageRankDef}) with 
$\mathbf{r}^{(b)}(\mathbf{P},\mathbf{s},\zeta,\beta)$. 
}

Let us first consider a simplification of this 
walk where the random walker after going back one step does not
step back further (his momwentary $\beta$ drops to zero). 
\footnote{
We will denote the stationary distribution under this behaviour as 
(\ref{eqGenLazyPageRankDef}) with 
$\mathbf{r}^{(b_1)}(\mathbf{P},\mathbf{s},\zeta,\beta)$. 
}
Under RBS settings the PageRank  of a node $j$ gets an authority of say  $p_j$ from its ``parents''
(that is: the parents in the network and any node that is jumped from upon walker getting bored) and an authority of say $c_j$ from its ``children'' (i.e. the children in the network and any node that is jumped to upon walker getting bored at $j$). Of course $r_j^{(b_1)}=p_j+c_j$. In the next step $\beta p_j$ is given away to ``parents'' of $j$ by backstep,
while the ``children'' get then $(1-\beta)p_j+c_j$, so that the ``children'' give back again $\beta((1-\beta)p_j+c_j)$. Upon stationary distribution we must have:
$c_j=\beta((1-\beta)p_j+c_j)$, hence
$(1-\beta)c_j=\beta (1-\beta)p_j $, and finally
$c_j=\beta p_j$, so that
$r_j^{(b_1)}=p_j+c_j=p_j+\beta p_j=(1+\beta)p_j$.
Hence, $c_j=\frac{\beta}{1+\beta}  r_j^{(b_1)}$.
This means the ``children'' get

$$ (1-\beta)p_j+c_j = (1-\beta)p_j+\beta p_j
=  p_j=\frac{1}{1+\beta}r_j^{(b_1)}$$

\noindent Out of this amount $\zeta r_j^{(b_1)}$ is distributed all over the network by boring jump, while the remaining authority  is assigned to real children 
$\frac{1}{1+\beta}r_j^{(b_1)}- \zeta r_j^{(b_1)}
=\left(\frac{1}{1+\beta}-\zeta\right) r_j^{(b_1)}
$
  in a walk.

Summarizing

$$\mathbf{p}=
     \mathbf{P} {\cdot}
\left(\frac{1}{1+\beta}-\zeta\right)\mathbf{r}^{(b_1)}
+  \zeta  \mathbf{s}$$

\noindent and

$$\mathbf{c}=\frac{\beta}{1+\beta}  \mathbf{r}^{(b_1)}
$$

\noindent Hence the equation below:

\begin{equation}
\mathbf{r}^{(b_1)}=\mathbf{c}+\mathbf{p}=
\frac{\beta}{1+\beta}  \mathbf{r}^{(b_1)}
+ 
\left(\frac{1}{1+\beta}-\zeta\right)  \mathbf{P} {\cdot}\mathbf{r}^{(b_1)}
+    \zeta \mathbf{s}
\label{eqRBSDef}
\end{equation}
 
But let us transform equation (\ref{eqRBSDef}):
$$
\mathbf{r}^{(b_1)}=
 \left(\frac{\beta}{1+\beta}  \mathbf{I} +
\left(\frac{1}{1+\beta}-\zeta\right) \mathbf{P} \right) {\cdot}
\mathbf{r}^{(b_1)} 
+  \zeta  {\cdot}  \mathbf{s}
$$
 
$$
\mathbf{r}^{(b_1)}=
(1-\zeta) 
 \left(\frac{\beta}{(1+\beta)(1-\zeta)}  \mathbf{I} +
\left(\frac{1}{(1+\beta)(1-\zeta)}-\frac{\zeta}{1-\zeta}\right) \mathbf{P} \right) 
{\cdot}
\mathbf{r}^{(b_1)}
+  \zeta   {\cdot}  \mathbf{s}
$$

\noindent Comparing this to the equation (\ref{eqGenLazyPageRankDef}) we see that this is a special case of generalized lazy random walk PageRank.
That is:
$
\mathbf{r}^{(b_1)}(\mathbf{P},\mathbf{s},\zeta,\beta)
=
\mathbf{r}^{(g)}(\mathbf{P},\mathbf{s},\zeta,\frac{\beta}{(1+\beta)(1-\zeta)} )$.

Now transform  the equation (\ref{eqRBSDef}) differently:

$$
\left(1-\frac{\beta}{1+\beta} \right)\mathbf{r}^{(b_1)}=
\left(\frac{1}{1+\beta}-\zeta\right)  \mathbf{P} {\cdot}\mathbf{r}^{(b_1)}
+    \zeta \mathbf{s}
$$

$$
 \frac{1}{1+\beta}  \mathbf{r}^{(b_1)}=
\left(\frac{1}{1+\beta}-\zeta\right)  \mathbf{P} {\cdot}\mathbf{r}^{(b_1)}
+    \zeta \mathbf{s}
$$

$$
   \mathbf{r}^{(b_1)}=
\left(1-\zeta(1+\beta)\right)  \mathbf{P} {\cdot}\mathbf{r}^{(b_1)}
+    \zeta(1+\beta) \mathbf{s}
$$

\noindent 
That is:
$
\mathbf{r}^{(b_1)}(\mathbf{P},\mathbf{s},\zeta,\beta)
=
\mathbf{r}^{(t)}(\mathbf{P},\mathbf{s},\zeta(1+\beta) )$. 
meaning that the RBS with single backstep  can be quite well reflected by traditional PageRank.
  This would imply that 
two theorems \ref{thPRlimtPreferential} and \ref{thPRlimit}
are effective not only for traditional but also for random walk with (single) backstep via slight change to the $\zeta$ coefficient.  
But we have to note   that we in fact assume here a limitation
on the probability of 
  a single back-step: $\frac{1}{1+\beta}>\zeta$. 

If we allow for multiple backsteps in row,   not only $p_j$ is subject to 
backstep, but also $c_j$. 
So  a node $j$,
with authority  $r_j^{(b)}=p_j+c_j$, 
in the next step gives away to "parents"  authority of $\beta (p_j+c_j)$  by backstep.

Its "children"
  get then 
$(1-\beta)r_j^{(b)}$.

Subsequently 
  the "children" give back again 
$\beta(1-\beta)r_j^{(b)}$.
But only this? No. 
While the authority passes down the children, 
eventual backstep may occur many steps away, providing probability masses of  
$\beta^2(1-\beta)^2r_j^{(b)}$, 
$\beta^3(1-\beta)^3r_j^{(b)}$, ... .
Most of these masses will occur multiple times because of 
different interleaving of steps forward and backward. 
To imagine how it works think of a Pascal triangle for the series of expressions 
$\left((1-\beta)+\beta\right)^n$ where 
$\beta$ mass passes down on the first level to the right branch. Going to the right and walking there downwards 
passing the symmetrical is just the backward return point
(same number of moves to the left and to the right standing for forward and backward moves).
Given that $\beta<1-\beta$ that is $\beta<0.5$,
the amount of mass to the right of the symmetrical will drop down to zero, meaning that walks forward/backward 
with forward being up to the last move frequent than  backward will move to the symmatral (nearly) the whole mass $\beta$. 

The other argument is that of stable state where the mass lost for "parents" has to be provided by the "children".  

Upon stationary distribution we must have then  
$c_j=\beta r_j $, hence 

Summarizing

$$\mathbf{p}=
     \mathbf{P} {\cdot}
\left(\frac{1}{1+\beta}-\zeta\right)\mathbf{r}^{(b_1)}
+  \zeta  \mathbf{s}$$

\noindent and

$$\mathbf{p}=
  (1-\zeta-\beta){\cdot}  \mathbf{P} {\cdot}
 \mathbf{r}^{(b)}
+   \zeta {\cdot}  \mathbf{s}
$$ 
and 
$$\mathbf{c}= \beta  \mathbf{r}^{(b)} 
$$
 
  Hence the equation below.  
\begin{equation}
\mathbf{r}^{(b)}=\mathbf{c}+\mathbf{p}=
\beta  \mathbf{r}^{(b)} 
+  (1-\zeta-\beta){\cdot}  \mathbf{P} {\cdot}
 \mathbf{r}^{(b)}
+   \zeta {\cdot}  \mathbf{s}
\label{eqRBSDefTwo}
\end{equation}
which is again easily translated to generalized lazy walk and traditional random walk, i.e. 
$
\mathbf{r}^{(b)}(\mathbf{P},\mathbf{s},\zeta,\beta)
=
\mathbf{r}^{(t)}(\mathbf{P},\mathbf{s},\frac{\zeta}{1- \beta} )$. 
This would imply that 
two theorems \ref{thPRlimtPreferential} and \ref{thPRlimit}
are effective not only for traditional but also for random walk with (multiple) backstep under modification of $\zeta$.

It is, however, to be remembered that $\beta<\frac12$. 
This is not surprising because you cannot go backward more often than you go forward in a browser. 
 
Note that in the paper \cite{Huang:2006:WCI} 
also some kind of random walk with "backstep" is considered, while the formulas the authors come at are apparently simpler (closed-form solutions).
However, we shall note here two things: 
(1) they concentrate there on walks with one step forward followed by one step backward, but (2) 
their step backward is not the intrinsic usage of backstep bottom by an internaut because on backstep they do not return to the previous page but rather to any of the inlinking pages. 
So our analysis does not apply to their setting nor their to ours.

\section{Bipartite PageRank}

Some non-directed graphs occurring e.g. in social networks
are in a natural way bipartite graphs. That is there exist nodes of two modalities and meaningful links may occur only between nodes of distinct modalities (e.g. clients and items purchased by them\Bem{, Fig.\ref{dwudzielny}}).

Some literature exists already for such networks attempting to adapt PageRank to the specific nature of bipartite graphs, e.g. \cite{DLK09}. Whatever investigations were run, apparently no generalization of theorem \ref{thPRlimit} was formulated.

One seemingly obvious choice would be to use the traditional PageRank, like it was done in papers \cite{Link:2011,Bauckhage:2008}. But this would be conceptually wrong because the nature of the super-node would cause authority flowing between nodes of the same modality which is prohibited by the definition of these networks.

Therefore in this paper we intend to close this conceptual gap
using  Bipartite PageRank concept created in our former paper \cite{Bipartite:2012}  and will extend the Theorem \ref{thPRlimit} to this case.

So let us consider the flow of authority in a bipartite network with two distinct super-nodes: one collecting the authority from items and passing them to clients, and the other the authority from clients  and passing them to items.

\begin{equation}
\mathbf{r}^p=(1-\zeta^{kp}){\cdot}\mathbf{P}^{kp}{\cdot} \mathbf{r}^k + \zeta^{kp}{\cdot}\mathbf{s}^p \label{eqBPRitem}
\end{equation}
\begin{equation}
\mathbf{r}^k=(1-\zeta^{pk}){\cdot}\mathbf{P}^{pk}{\cdot} \mathbf{r}^p+\zeta^{pk}{\cdot}\mathbf{s}^k \label{eqBPRclient}
\end{equation}

The following notation is used in these formulas
\begin{itemize}
\item $\mathbf{r}^p$, $\mathbf{r}^k$, $\mathbf{s}^p$, and $\mathbf{s}^k$ are stochastic vectors, i.e. the non-negative elements of these vectors sum to 1;
\item the elements of matrix $\mathbf{P}^{kp}$ are: if there is a link from page $j$ in the set of $Clients$ to a page $i$ in the set of $Items$, then $p^{kp}_{ij}=\frac{1}{outdeg(j)}$, otherwise $p^{kp}_{ij}=0$;
\item the elements of matrix $\mathbf{P}^{pk}$ are: if there is a link from page $j$ in the set of $Items$ to page $i$ in the set of $Clients$, then $p^{pk}_{ij}=\frac{1}{outdeg(j)}$, and otherwise $p^{pk}_{ij}=0$;
\item $\zeta^{kp}\in [0,1]$ is the boring factor when jumping from \emph{Clients} to \emph{Items};
\item $\zeta^{pk}\in [0,1]$ is the boring factor when jumping from {Items} to {Clients}.
\end{itemize}

\Bem{
\begin{figure}
   \centering
   \includegraphics[width=0.9\textwidth]{\rys{}newtheorem.jpg}\\
\caption{Two sets of preferred nodes in a bipartite graph}
\label{newtheorem}
\end{figure}

}

\begin{definition}
The solutions $\mathbf{r}^p$ and $\mathbf{r}^k$ of the equation system (\ref{eqBPRitem})  and (\ref{eqBPRclient}) will be called item-oriented and client-oriented bipartite PageRanks,  resp.
\end{definition}

Let us assume first that

\[\zeta^{pk}=\zeta^{kp}=0\]

\noindent i.e. that the super-nodes have no impact.

Let $K=\sum_{j\in Clients}outdeg(j)=\sum_{j\in Items}outdeg(j)$ mean the number of edges leaving one of the modalities.
Then for any $j\in Clients$ we have $r^{k}_{j}=\frac{outdeg(j)}{K}$, and for any $j\in Items$ we get $r^{p}{j}=\frac{outdeg(j)}{K}$. Because through each channel the same amount of $\frac{1}{K}$ authority is passed,  within each bidirectional link the amounts passed cancel out each other. So the $\mathbf{r}$'s defined this way are a fix-point (and solution) of the equations (\ref{eqBPRitem})  and (\ref{eqBPRclient}).

For the other extreme, when $\zeta^{kp}=\zeta^{pk}=1$ one obtains, that $\mathbf{r}^p=\mathbf{s}^p$, $\mathbf{r}^k=\mathbf{s}^k$.

In analogy to the traditional PageRank let us note at this point that for $\zeta^{kp}, \zeta^{pk}>0$ the ``fan''-nodes of both the modalities (the sets of them being denoted with $U^p$ for items and $U^k$ for clients), will obtain in each time step from the super-nodes the amount of authority equal to  $\zeta^{pk}$ for clients and $\zeta^{pk}$ for products, resp.

Let us now think about a fan of the group of nodes $U^p, U^k$ who jumps uniformly,  Assume further that at the moment $t$
we have the following state of authority distribution:
node $j$ contains $r^{k}_{j}(t)=\frac{1}{|U^k| }, r^{p}_{j}(t)=\frac{1}{|U^p| }$
(meaning analogous formulas for $r^p$ and $r^k$). Let us consider now the moment $t{+}1$. From the product node $j$ to the first super-node the authority $\zeta^{pk} \frac{1}{|U^p| }$ flows, and into each outgoing link $(1-\zeta^{pk}) \frac{1}{|U^p| deg(j)}$ is passed. On the other hand the client node $c$ obtains from the same super-node authority $\zeta^{pk} \frac{1}{|U^k| }$,
while from link   ingoing from j $(1-\zeta^{pk})  \frac{1}{|U^p| deg(j)}$.
the authority from clients to products passes in the very same way.

We have a painful surprise this time. In general we cannot define a useful state of authority of nodes, analogous to that of traditional PageRank from the previous section,  so that in both directions between $U^p$ and $U^k$ nodes the same upper limit of authority would apply.  This is due to the fact that in general capacities of $U^k$ and $U^p$ may differ. Therefore a broader generalization is required.

To find such a generalization let us reconsider the way how we can limit the flow of authority in a single channel. The amount of authority passed consists of two parts: a variable one being a share of the authority at the feeding end of the channel and a fixed one coming from a super-node. So, by increasing the variable part we come to the point that the receiving end gets less authority that was there on the other end of the channel.
 
Let us seek the amount of authority $d$ such that multiplied by the number of out-links of a sending node will be not lower than the authority of this node and that after the time step
its receiving node would have also amount of authority equal or lower than $d$ multiplied by the number of its in-links.
That is we want to have that:

$$ d{\cdot}(1-\zeta^{pk}) + \frac{\zeta^{pk}}{\sum_{v \in U^k}outdeg(v)}  \le d $$

The above relationship corresponds to the situation 
that 
on the one hand if a node in $Items$ has at most $d$ amount of authority per link, then it sends to a node in $Clients$ 
at most $d{\cdot}(1-\zeta^{pk})$ authority via the link .
The receiving node $j$ on the other hand, if it belongs to 
$U^k$, then it gets additionally from the supernode exactly
$\frac{\zeta^{pk}}{| U^k| deg(j)}$  authority per its link.
We seek a $d$ such that these two components do not exceed $d$ together. 

If we look from the perspective of pasing authority from $Clients$ to $Items$, then, for similar reasons  at the same time we have 
$$d{\cdot}(1-\zeta^{kp}) + \frac{\zeta^{kp}}{| U^p|deg(j)} \le d  $$

This implies immediately, that

$$d\ge \frac{1}{| U^k|\min_{j\in U^k}deg(j)} $$
and
$$d\ge \frac{1}{| U^p|\min_{j\in U^p}deg(j)}
$$
so we come to a satisfactory $d$ when  
$$d=\max(\frac{1}{| U^k|\min_{j\in U^k}deg(j)}
,       \frac{1}{| U^p|min_{j\in U^p}deg(j)})
$$ $$
=\frac{1}{\min(| U^k|\min_{j\in U^k}deg(j)
,     | U^p|min_{j\in U^p}deg(j))}
$$

Now we are ready to formulate a theorem for bipartite PageRank
analogous to the preceeding theorem~\ref{thPRlimit}.

\begin{theorem} \label{thBPRlimit}
For the uniform personalized bipartite PageRank we have
$$p_{k,o}\zeta^{kp}\le 
\frac{(1-\zeta^{pk}) \partial(\frac{U^p}{U^k}) }{min(| U^k|min_{j\in U^k}deg(j)
,     | U^p|min_{j\in U^p}deg(j))}
$$
and
$$p_{p,o}\zeta^{pk}\le 
\frac{(1-\zeta^{kp}) \partial(\frac{U^k}{U^p}) }{min(| U^k|min_{j\in U^k}deg(j)
,     | U^p|min_{j\in U^p}deg(j))}
$$
where
\begin{itemize}
\item $p_{k,o}$ is the sum of authorities from the set $Clients\backslash U^k$,
\item $p_{p,o}$ is the sum of authorities from the set $Items \backslash U^p$,
\item $\partial(\frac{U^k}{U^p})$ is the set of edges outgoing from $U_k$ into nodes from $Items-U_p$ (that is ``fan's border'' of $U^k$),
\item $\partial(\frac{U^p}{U^k})$ is the set of edges outgoing from $U^p$ into nodes from $Clients \backslash U^k$ (that is ``fan's border'' of $U^p$),
\end{itemize}
\qed
\end{theorem}

The proof  is analogous as in case of classical PageRank, using now the quantity $d$ we have just introduced.

\begin{proof}
Let us notice first that, due to the closed loop of authority circulation, the amount of authority flowing into $U^k$ from the nodes belonging to the set $\overline{U^p} = Items \backslash U^p$ must be identical with the amount flowing out of $U^p$ to the nodes~in~$\overline{U^k}$.
The same holds when we exchange the indices $p<->k$. 

But from $U^p$ only that portion of authority flows out to
$\overline{U^k}$
  that flows out through the boundary of $U^p$ because no authority leaves the tandem  $U^p,U^k$ via super-nodes (it returns from there immediately). As the amount $d  |\partial(\frac{U^p}{U^k})| $ leaves at most the $U^p$ not going into $U^k$, then

$$p_{k,o}\zeta^{kp}\le  d (1-\zeta^{pk}) 
\partial(\frac{U^p}{U^k}) = $$ $$
=\frac{(1-\zeta^{pk})
 \partial(\frac{U^p}{U^k}) }{min(| U^k|min_{j\in U^k}deg(j)
,     | U^p|min_{j\in U^p}deg(j))}
 $$
\end{proof}

\Bem{
Note that we can alternatively derive separate upper limits $d^k$ and $d^p$ for the clients' and  items' sub-graphs

$$\Big[d^k{\cdot}(1-\zeta^{kp}) + \frac{\zeta^{kp}}{\sum_{v\in U^p} outdeg(v)}\Big] \cdot (1-\zeta^{pk}) +
\frac{\zeta^{pk}}{\sum_{v \in U_k}outdeg(v)} \le d^k $$

$$\Big[d^p{\cdot}(1-\zeta_{pk}) + \frac{\zeta^{pk}}{\sum_{v \in U^k} outdeg(v)}\Big]\cdot (1-\zeta^{kp}) + \frac{\zeta^{kp}}{\sum_{v \in U^p}outdeg(v)} \le d^p $$
}

One topic was not touched above, namely that of convergence.
But the convergence can be looked for in an analogous way as done for the HITS (consult e.g. \cite[Ch. 11]{LM06}).

\section{Experimental exploration of the limits}

With the established limits, we can pose now the question 
how tight the limits are or rather whether we can construct networks for which the limits are approached sufficiently closeLet us first look at some small examples.

 \begin{figure}
   \centering
   \includegraphics[width=0.9\textwidth]{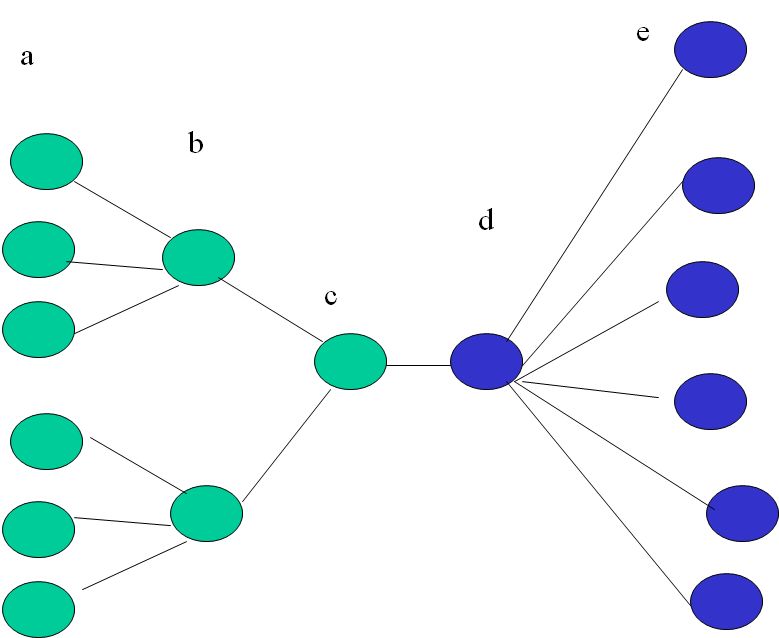}\\
\caption{An unoriented tree-like network}
\label{siecvDrzewo}
\end{figure}
 
 \begin{figure}
   \centering
   \includegraphics[width=0.9\textwidth]{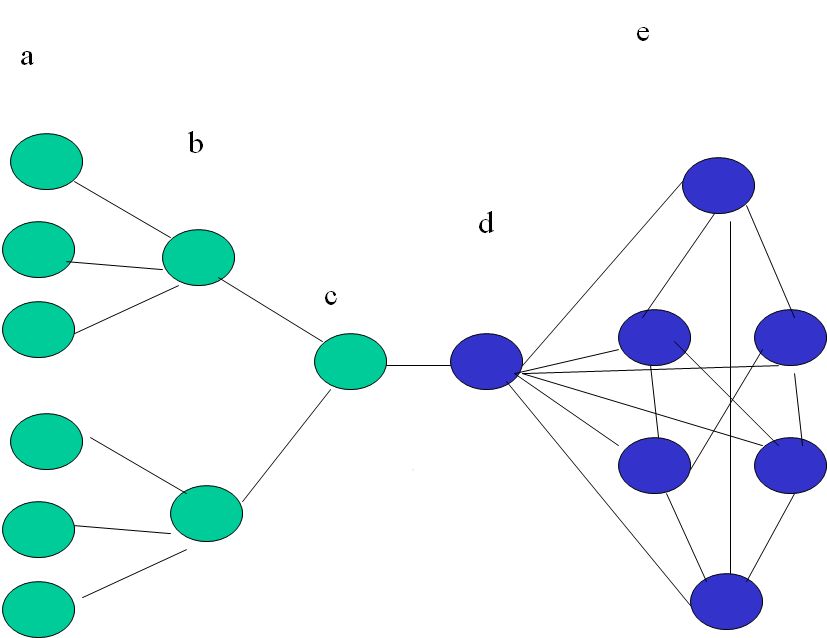}\\
\caption{An unoriented complex network}
\label{siecvZlvozxona}
\end{figure}

For this purpose we will use a family of networks depicted in figures 
\ref{siecvZlvozxona} and \ref{siecvDrzewo}. 
The network is divided into three "zones" of nodes. 
Zones d and e belong to the set of fan-nodes. 
Zones a,b,c are not fan-sets. 
There is only one node in zones c and d
so that the edge connecting d to c is the channel through which the authority flows out of the fan-node set
and we seek the upper limit of authority lost via this link.
The zones are symmetricvally constructed.
The number of nodes in a is a multiple of the numvber of nodes in b. 
All nodes in e are connected to d and otherwise they constitute a regular subgraph. 
In figure  \ref{siecvDrzewo} this subgraph is of degree zero, 
and in  \ref{siecvZlvozxona} it is of degree 3.

Because of symmetry the PageRanks in each of the zones are identical. 

Table \ref{siecvDrzewoPR} shows the PageRanks for the graph in Figure
\ref{siecvDrzewo}.  
Table \ref{siecvZlvozxonaPR} shows the PageRanks for the graph in Figure
\ref{siecvZlvozxona}.  
In each table the columns "zone a",...,"zone e" show the PageRank attained by each node in the respective zone. 
'outflow" column shows the amount of authority flowing out from the fan-set of nodes to the rest of the network. 
"limit" column is the upper limit derived theoretically in the previous sections for the respective case. 
"rel.left" is cmputed as 1-"outflow"/"limit". The lower the value the closer the actual outflow to the theoretical limit.

\begin{table}
\caption{PageRanks for network Fig.\ref{siecvDrzewo}. Boring factor=0.1 }
\label{siecvDrzewoPR}
{
{\footnotesize
\begin{tabular}{p{3.7cm}|lllll}
& zone a & zone b & zone c & zone d & zone e   \\ \hline  \\
traditional uniform 
&   0.012479 
&   0.055464 
&   0.072565 
&   0.370274 
&   0.061892 \\
traditional preferential 
&   0.013019 
&   0.057864 
&   0.075705 
&   0.386296 
&   0.057358   \\ \hline   \\
 & outflow & limit
& rel.left \\ \hline  \\
traditional uniform 
&   0.025837 &  0.12857142 & 0.799\\
traditional preferential 
&   0.026955 &   0.069230 & 0.610  \\
\end{tabular}
}
}
\end{table}
 
\begin{table}
\caption{PageRanks for network Fig.\ref{siecvZlvozxona}. Boring factor=0.1 
}\label{siecvZlvozxonaPR}
{\footnotesize
\begin{tabular}{p{3.7cm}|llllllll}
& zone a & zone b & zone c & zone d & zone e  \\ \hline  \\
traditional uniform 
&    0.006094806 & 0.027088026 & 0.035440167 & 0.180837655  &0.115496215
  \\
traditional preferential 
&  0.006306085 & 0.028027043 & 0.036668714 & 0.187106461 & 0.113722372
\\ \hline   \\
%
 & outflow & limit
& rel.left \\ \hline  \\
traditional uniform 
&  0.0126185 & 0.032142 & 0.6074242 \\
traditional preferential 
&  0.013055 1& 0.029032 &  0.5502957 \\
\end{tabular}

}
\end{table}
 
The obvious tendency to keep authority is observed when then network of connections
is densified between fan nodes. 
Also the outflow of authority gets closer to the theoretical bound. 

How close it can go?
In tables
 \ref{siecvZlvozxonaPR1000p}
and
 \ref{siecvZlvozxonaPR1000u}
 we increase by the factor of 10,100 etc. the number of nodes
in zones a,b and e and also the number of connections between 
the nodes in zone e (enlarging the network of fig.\ref{siecvZlvozxona}).

\begin{table}
\caption{PageRanks for enlarged network Fig.\ref{siecvZlvozxona} by factor in the first column. Boring factor=0.1. Traditional PageRank with preferential authority re-distribution.} \label{siecvZlvozxonaPR1000p}
{\footnotesize 
\begin{tabular}{l|llllllll}
factor& zone a & zone b & zone c & zone d & zone e  \\ \hline  \\
10 
&     1.737833e-05 & 7.723703e-05 & 7.073624e-04 &2.438616e-02 &1.620532e-02
\\ 100 & 
 1.969998e-08 & 8.755548e-08 &7.674967e-06 &2.494130e-03 &1.662448e-03
\\ 1000 & 
 1.995495e-11 & 8.868865e-11 &7.739489e-08 &2.499416e-04 &1.666249e-04
\\ 10000 & 
 1.998069e-14 & 8.880307e-14 &7.745988e-10 &2.499942e-05 &1.666625e-05
\\ 100000 & 
 1.998323e-17 & 8.881436e-17 &7.746628e-12 &2.499994e-06 &1.666662e-06
\\ \hline   \\
%
factor & outflow & limit
& rel.left \\ \hline  \\
10 
 &0.0003294802  &  0.0003657049979 & 0.0990544634
\\ 100 & 
 3.700605208e-06 & 3.740632831e-06 & 0.01070076246
\\ 1000 & 
  3.745018664e-08&3.749062578e-08& 0.0010786466
\\ 10000 & 
  3.749501576e-10&3.749906250e-10& 0.000107915 
\\ 100000 & 
  3.749943936e-12&3.749990625e-12& 1.245027547e-05
\end{tabular}
}
\end{table}

\begin{table}
\caption{PageRanks for enlarged network Fig.\ref{siecvZlvozxona} by factor in the first column. Boring factor=0.1. Traditional PageRank with uniform authority re-distribution.} \label{siecvZlvozxonaPR1000u}
{\footnotesize 
\begin{tabular}{l|llllllll}
factor& zone a & zone b & zone c & zone d & zone e  \\ \hline   
\\ 10  &
  1.679439e-05 & 7.464174e-05  & 6.835939e-04  &  2.356675e-02  &  1.622082e-02
\\ 100  &
1.904272e-08  & 8.463432e-08  & 7.418903e-06  & 2.410918e-03  & 1.662589e-03
\\ 1000  &
 1.928972e-11 &  8.573209e-11  & 7.481482e-08  & 2.416095e-04  & 1.666263e-04
\\ 10000  &
 1.931466e-14 &  8.584294e-14  & 7.487786e-10  & 2.416610e-05 &  1.666626e-05
\\ 100000  &
 1.931710e-17  & 8.585376e-17  & 7.488401e-12  & 2.416661e-06  & 1.666663e-06
\\ 1000000  &
 1.931896e-20  & 8.586206e-20  & 7.488419e-14  & 2.416666e-07  & 1.666666e-07
\\ \hline   \\
%
%
factor&  outflow & limit
& rel.left \\ \hline   
\\ 10  &
  0.0003184092239& 0.0003688524590  &    0.1367572
\\ 100  &
   3.577140044e-06 & 3.743760399e-06  &  0.04450614810 
\\ 1000  &
   3.620173279e-08 & 3.749375104e-08  &  0.03445956209 
\\ 10000  &
   3.624516929e-10 & 3.749937501e-10  &  0.03344604329 
\\ 100000  &
  3.624940904e-12 & 3.749993750e-12  &  0.033347481127
\\ 1000000  &
   3.625220796e-14 & 3.74999937e-14  &  0.033274293139 
\end{tabular}
}
\end{table}

We see that in case of preferential attachment we quickly approach the bounds. 
In case of uniform authority redistribution we get a stabilization. 

The situation changes for uniform case, however, if we densify the connections in zone e.
For the network of the last line we increase the densiy of connectiions uin zone e.

\begin{table}
\caption{PageRanks for densified  network from last line of previous table - the zone e node degrees as in the first column}  \label{siecvZlvozxonaPR1000udens}
{\footnotesize
\begin{tabular}{l|llllllll}
e node deg. & zone a & zone b & zone c & zone d & zone e \\
\hline  
\\ 5000000  &
   1.572210e-20  & 6.987602e-20  &  6.094048e-14  &  1.966666e-07  & 1.666666e-07
\\5500000&
   1.440820e-20  & 6.403644e-20  & 5.585813e-14   & 1.803030e-07   & 1.666666e-07
\\ 5900000&
   1.352272e-20  & 6.010099e-20  &  5.242307e-14 & 1.692090e-07   & 1.666666e-07
\\  5990000&
  1.333886e-20  & 5.928383e-20   & 5.171155e-14   & 1.669171e-07   & 1.666666e-07
\\ 5999000 &
  1.331915e-20  & 5.919623e-20   & 5.163832e-14   & 1.666916e-07   & 1.666666e-07
\\5999900&
 1.332161e-20  & 5.920716e-20   & 5.163986e-14   & 1.666691e-07   & 1.666666e-07
\\ \hline   \\
%
%
e node deg. &   outflow & limit
& rel.left \\
\hline  
\\ 5000000  &
   2.950251336e-14   &  2.999999500e-14   &  0.01658272402
\\5500000&
     2.703801851e-14 & 2.727272272e-14   &  0.0086058227330
\\ 5900000&
      2.537613978e-14 & 2.542372457e-14   &  0.00187166895
\\  5990000&
    2.503123889e-14 & 2.504173205e-14   &  0.00041902692587
\\ 5999000 &
    2.4994569e-14 & 2.500416319e-14   &  0.0003836900900
\\5999900&
   2.49983829e-14   &  2.500041250e-14   &  8.117929671e-05
\end{tabular}
}
\end{table}

Last not least let us observe that the relationship between 
the upper limit and the actual amount of authority passed 
is a function of the structure of the network.
In the tables  \ref{siecvZlvozxonaPR1000sspref} (for preferential redistribution)
and \ref{siecvZlvozxonaPR1000ssuni} (for uniform redistribution)
we see this effect.
For preferential redistribution we see that the lower degrees the nodes are, the bigger part of the authority is flowing out.  
For the uniform redistribution the tendency is in the other direction.

\begin{table}
\caption{PageRanks for various network structures 
with the same upper limit of authority passing - the preferential redistribution.
Zone a and b both 60000 nodes each.}
\label{siecvZlvozxonaPR1000sspref}  
{\footnotesize 
\begin{tabular}{l|llllllll}
e node  deg.  & zone a & zone b & zone c & zone d & zone e   \\
/ count \\
\hline
\\ 511 / 1024 & 
  6.092727e-11 & 1.353939e-10 & 5.370716e-06 & 1.953285e-03 & 9.746382e-04
\\ 255 / 2048&
  6.095403e-11 & 1.354534e-10 & 5.373075e-06 & 3.906379e-03 & 4.863655e-04
\\  127 / 4096 &  
 6.096742e-11 & 1.354832e-10 & 5.374255e-06 & 7.812567e-03 & 2.422291e-04
\\ 63 / 8192 &
  6.097412e-11 & 1.354980e-10 & 5.374845e-06 & 1.562494e-02 & 1.201609e-04
\\ 31 / 16384 &
  6.097746e-11 & 1.355055e-10 & 5.375140e-06 & 3.124970e-02 & 5.912678e-05
\\ 15 /   32768 &
  6.097914e-11 & 1.355092e-10 & 5.375287e-06 & 6.249920e-02 & 2.860973e-05
\\ 7 /65536&  
  6.097997e-11 &  1.355110e-10 & 5.375361e-06 & 1.249982e-01 & 1.335121e-05
\\3  // 131072&
 6.098038e-11  & 1.355119e-10 & 5.375397e-06 & 2.499961e-01 & 5.721944e-06
\\ 1 / 262144 &
  6.098056e-11 & 1.355124e-10 & 5.375413e-06 & 4.999919e-01 & 1.907314e-06
\\ \hline   \\
%
e node  deg.   &   outflow & limit& rel.left \\
/ count \\
\hline
\\ 511 / 1024 & 
   1.714998913e-06 & 1.716610495e-06 & 0.0009388162095
\\ 255 / 2048&
    1.715752081e-06 &  1.716610495e-06 &  0.0005000635487 
\\  127 / 4096 &  
   1.716128933e-06 &  1.716610495e-06 &  0.0002805306660
\\ 63 / 8192 &
   1.71631741e-06 &  1.716610495e-06 &  0.0001707321858 
\\ 31 / 16384 &
   1.716411644e-06 & 1.716610495e-06  & 0.0001158392913
\\ 15 /   32768 &
    1.716458707e-06 &  1.716610495e-06 &  8.842327961e-05 
\\ 7 /65536&  
    1.716482125e-06 &  1.716610495e-06 &  7.478124133e-05 
\\3  // 131072&
   1.716493600e-06 &  1.716610495e-06 &  6.809630427e-05 
\\ 1 / 262144 &
    1.716498849e-06 &  1.716610495e-06 &  6.503878422e-05 
 \end{tabular}

}
\end{table}

\begin{table}
\caption{PageRanks for various network structures 
with the same upper limit of authority passing - the uniform redistribution  
Zone a and b both 60000 nodes each.} \label{siecvZlvozxonaPR1000ssuni}
{\footnotesize 
\begin{tabular}{l|llllllll}
e node  deg.  & zone a & zone b & zone c & zone d & zone e   \\
/ count \\
\hline  
\\ 4 / 131071 &
  4.480253e-11 & 9.956117e-11 & 3.949326e-06 & 1.836718e-01 & 6.228041e-06
\\8 / 65535 &   
4.933356e-11 & 1.096301e-10 & 4.348734e-06 & 1.011236e-01 & 1.371576e-05
\\ 16 /  32767 &
5.196287e-11 & 1.154730e-10 & 4.580507e-06 & 5.325655e-02 & 2.889275e-05
\\32 / 16383 &
 5.339301e-11 & 1.186511e-10 & 4.706573e-06 & 2.736115e-02 & 5.936787e-05
\\64 / 8191 &
5.416866e-11 & 1.203748e-10 & 4.774947e-06 & 1.387931e-02 & 1.203889e-04
\\ 128 / 4095 &
  5.468880e-11 & 1.215307e-10 & 4.820796e-06 & 7.006293e-03 & 2.424855e-04
\\ 256 / 2047 &
  5.545008e-11 & 1.232224e-10 & 4.887903e-06 & 3.551911e-03 & 4.867770e-04
\\ 512 / 1023 &
  5.783015e-11 & 1.285114e-10 & 5.097705e-06 & 1.852184e-03 & 9.756907e-04
\\ \hline   \\
%
e node  deg.   &   outflow & limit& rel.left \\
/ count \\
\hline  
\\ 4 / 131071 &
    1.261114759e-06 &  1.716613769e-06 &  0.2653474055 
\\8 / 65535 &   
   1.388655573e-06 &  1.716613769e-06 &  0.1910494961 
\\ 16 /  32767 &
  1.462666077e-06 &  1.716613769e-06 &  0.147935252
\\32 / 16383 &
    1.502922183e-06 &  1.716613769e-06 &  0.1244843712 
\\64 / 8191 &
   1.524755444e-06 &  1.716613769e-06 &  0.1117655749
\\ 128 / 4095 &
    1.539396343e-06 &  1.716613769e-06 &  0.1032366329 
\\ 256 / 2047 &
   1.560825131e-06 &  1.716613769e-06 &  0.09075345911 
\\ 512 / 1023 &
    1.627819998e-06 &  1.716613769e-06 &  0.05172612086 
 \end{tabular}

}
\end{table}

\section{Concluding Remarks}

In this paper we have  proposed limits for
the flow of  authority in ordinary undirected and in  bipartite graph under uniform random jumps.
We have empirically demonstrated tightness of some of these limits. 

For the ordinary indirected graphs 
we have considered five versions of random walkers for computation of PageRank. Each of them represents semantically different behavior of the surfers and hence the respective PageRank has different commercial values as placement if e.g. advertisement is concerned.

The obtained limits can be used for example when verifying validity of clusters in such graphs. It is quite common to assume that the better the cluster the less authority flows out of it when treating the cluster as the set on which a fan concentrates while a personalized PageRank is computed. The theorem says that the outgoing authority has a natural upper limit dropping with the growth of the  size of the sub-network so that the outgoing authority cluster validity criterion cannot be used because it will generate meaningless large clusters. So a proper validity criterion should make a correction related to the established limits in order to be of practical use.

This research needs to be seen in the broader context of our research efforts. 
Our group \Bem{Institute of Computer Science} is engaged
in developing a semantic search engine covering the whole Polish Internet.
So let us briefly explain the notion of semantic search engine and its impact on ranking method requirements.

Semantics  of information\footnote{Information is the content of a message  sent by a sender  to a recipient  in order to increase the level of knowledge of the recipient. As the sender may not always be aware of the level of recipient  knowledge,
one introduces the so-called intentional information. It means the perception of the information at the source (sender), that is the expected increase in recipient knowledge given the assumption of recipient knowledge and his capabilities to decode the message. } expresses the meaning of this information. In linguistics research on the semantics tries among others to  relate symbols (like words, phrases, characters) to (real) beings  which they mean (so-called denotations) therefore related areas like morphological and syntactic analysis is engaged. Understanding of semantics may prove useful in comprehending pragmatics (the expected acting of the recipient upon obtaining the information)
and apobetics (the goal of the sender when sending the information).

Identification of the meaning of an information has been subject of intense research. So-called ``semantic search'' is deemed to be a method of improvement of search engine response by means  of understanding of user intent as well as of the search terms in the context of the document space. If we take into account advances in natural language research, we easily guess that there is virtually no chance to realize the goal of semantic search, formulated in this way, in near future.
Computers have no chance to understand semantics of textual messages as they have no ``experience'' with the reality surrounding us humans. Access to semantics of real world appears to be a remote goal.

Therefore in our research project NEKST we reformulated in a significant way the task of semantic analysis of Internet documents by understanding the task in an operational way.
Instead of trying to pretend that the machine understands the meaning of the text, we use the fact that both the information sender and the recipient are human beings. Hence not the search engine but the man has to understand the text, and the search engine only supports him in this understanding.
This support has the form of so-called semantic transformations on the text which on the one hand enrich the text with new features extending search characteristics and on the other hand may move the text to other space than the document space that is into the space of objects the documents are about.

So the semantic transformation means such a transformation of the document and/or query content that allows for traditional document search via a semantically related query
 \cite{BEATCA:2010,BEATCA:2011}.

Within the system NEKST the following types of semantic transformations have been implemented:
\begin{itemize}
\item user suggestions,
\item substitution with synonyms, hypernyms, hyponyms  and other related concepts,
\item concept disambiguation,
\item document categorization,
\item personalized PageRank,
\item cluster analysis and assignment of cluster keywords to documents,
\item explicit separation of document cluster and document search,
\item extraction of named entities and relations between them,
\item diversification of responses to queries,
\item dynamic summarizing, and
\item identification and classification of harmful contents.
\end{itemize}

If you take the semantic transformation view then it is obvious that you need all the traditional mechanisms of a search engine also under semantic search, including the ranking as well as clustering mechanisms, because they are actually underpinning nearly all the mentioned semantic transformation.  
So in particular PageRank is to be considered from the emantic transformation point of view.
But also, as stated,  
PageRank itself is a career of semantic information. One usually assumes that a link is added to a page with some semantic relation to the pointed page in mind.
Various considerations brought about variants of PageRank that need to be considered from sewmantic point of view.

In this paper we were able to point at a unified view of a couple of PageRank variants and showed that the derived theorem on authority flow limits can be easily transferred between them. 

This means that one needs in fact only one version of Pagerank algortithm to capture various aspects of semantics and also to consider only one kind of authority flow limits when e.g. applying PageRank based custering methods.

However, it is necessary to take into account the goals of a semantic search engine as well as the developments on the network following the publication of ranking mechanism of PageRank, used then by many search engines.

As a further research direction it is obvious that finding tighter limits is needed. This would improve the evaluation of e.g. cluster quality.

\bibliographystyle{plain}
\bibliography{Trad1_PR_bib,Trad2_PR_bib}


\end{document}